\documentclass[conference]{IEEEtran}
\IEEEoverridecommandlockouts
\usepackage{cite}
\usepackage{amsmath,amssymb,amsfonts}
\usepackage{algorithmic}
\usepackage{amsthm}
\usepackage{graphicx}
\usepackage{textcomp}
\usepackage{setspace}
\usepackage{epsfig}
\usepackage{color}
\usepackage{lettrine}
\usepackage{mathtools}
\usepackage{float}
\usepackage{array}
\usepackage{multicol}
\usepackage[space]{grffile}
\usepackage{url}
\usepackage{mathtools}

\newcommand{\jj}{$\mathcal{FJ}$}
\newcommand{\rr}{$\mathcal{R}$}
\newcommand{\src}{$\mathcal{S}_1$}
\newcommand{\des}{$\mathcal{S}_2$}
\newcommand{\E}{\ensuremath{\mathbb E}}

\def \treq {\stackrel{\tiny \Delta}{=}}

\newtheorem{remark}{Remark}
\newtheorem{proposition}{Proposition}

\def\BibTeX{{\rm B\kern-.05em{\sc i\kern-.025em b}\kern-.08em
    T\kern-.1667em\lower.7ex\hbox{E}\kern-.125emX}}
\begin{document}

\title{Secure two-way communication via a wireless powered untrusted relay and friendly jammer}

\author{
	
	\IEEEauthorblockN{
		Milad Tatar Mamaghani\IEEEauthorrefmark{1}, Abbas Mohammadi\IEEEauthorrefmark{1}, Phee Lep Yeoh\IEEEauthorrefmark{2}, and Ali Kuhestani\IEEEauthorrefmark{1}}
	
	\IEEEauthorblockA{
		\IEEEauthorrefmark{1}Electrical Engineering Department, Amirkabir University of Technology (Tehran Polytechnic), Tehran, Iran}
	
	\IEEEauthorblockA{
		\IEEEauthorrefmark{2}School of Electrical and Information Engineering, The University of Sydney, NSW, Australia}

}

\maketitle

\begin{abstract}
In this paper, we propose a self-dependent two-way secure communication where two sources exchange confidential messages via a wireless powered untrusted amplify-and-forward (AF) relay and friendly jammer (FJ). By adopting the time switching (TS) architecture at the relay, the data transmission is accomplished in three phases: Phase I) Energy harvesting by the untrusted relay and the FJ through non-information transmissions from the sources, Phase II) Information transmission by the sources and jamming transmissions from the FJ to reduce information leakage to the untrusted relay; and Phase III) Forwarding the scaled version of the received signal from the untrusted relay to the sources. For the proposed system, we derive a new closed-form lower bound expression for the ergodic secrecy sum rate (ESSR). Numerical examples are provided to demonstrate the impacts of different system parameters such as energy harvesting time, transmit signal-to-noise ratio (SNR) and the relay/FJ location on the secrecy performance. The numerical results illustrate that the proposed network with friendly jamming (WFJ) outperforms traditional one-way communication and the two-way without friendly jamming (WoFJ) policy.
\end{abstract}


\section{Introduction}
Wireless physical-layer security (PLS) is a prominent paradigm for improving the information transmission security of future generation communications networks {\cite {review0}}, {\cite {review1}}. Due to the broadcast nature of wireless transmissions, PLS strategies are designed to exploit the fading channel dynamics of the legitimate users and potential eavesdroppers to support secure transmission. A key area of interest is the untrusted relaying scenario where the source-to-destination transmission is assisted by a relay which may also be a potential eavesdropper {\cite {review2}}. This scenario occurs in large-scale wireless systems such as heterogeneous networks and device-to-device (D2D) communication networks where confidential messages are often retransmitted by multiple intermediate nodes.

Secure transmission utilizing an untrusted relay was first studied in {\cite {Oohama}}, where an achievable secrecy rate was derived. In
{\cite {he1}}, it was found that introducing a friendly jammer (FJ) could result in a positive secrecy rate for a one-way untrusted relay
link with no direct source-destination transmission. Indeed, many recent papers on untrusted relay communications have
focused on the one-way relaying scenario {\cite {Kuhestani2}}, {\cite {cioffi}}. Recently, several works have considered the more interesting scenario of two-way untrusted relaying {\cite {Kuh3}}, {\cite {Green}} where physical-layer network coding can provide security enhancement since the relay receives a superimposed signal from the two sources instead of each individual signal \cite{Kushi}.

Due to the energy constraints of wireless networks, recent research have considered the use of wireless energy harvesting in PLS scenarios with relays and FJs {\cite {salman}}--{\cite {Zhao}}. To be specific, the authors in \cite{salman} employed a FJ with wireless energy harvesting to provide secure communication between a source and destination in the presence of an external eavesdropper. In \cite{Kalamkar}, the authors studied secure one-way communication via a wireless energy harvesting untrusted relay and provided a lower bound expression for the ergodic secrecy rate. In \cite{Zhao}, the secrecy rate of the wireless-powered relay networks was maximized by jointly designing power splitting and relay beamforming techniques.

In this paper, inspired by the work in {\cite {Kalamkar}}, {\cite {Zhao}}, we propose and analyze a new PLS scenario of wireless energy harvesting in two-way untrusted amplify-and-forward relaying in the presence of a FJ. Both the relay and FJ are powered by radio frequency (RF) signals from the two source nodes under a time switching (TS) policy. For the proposed system, we derive a new closed-form lower bound for the ergodic secrecy sum rate (ESSR). Numerical results show that the proposed two-way secure communication with friendly jamming (WFJ) provides significant security advantages compared to the one-way communication scheme proposed in \cite{Kalamkar} and two-way relaying without friendly jamming (WoFJ). We further discuss important design insights into the security impact of key system parameters including energy harvesting time ratio, transmit signal-to-noise ratio (SNR) and relay/FJ location.

\section{System Model and Signal Representation}

\begin{figure}[t]
	\centering
	\includegraphics[width= \columnwidth]{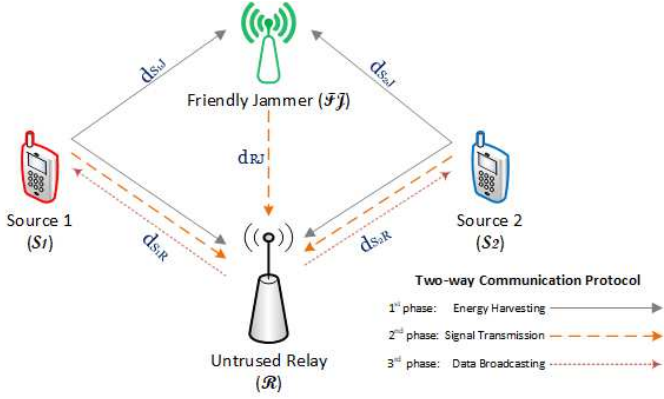}
	\caption{System model of two-way secure wireless powered network via an untrusted relay and a friendly jammer}
	\label{fig}
\end{figure}
We consider the PLS of a two-way communication scenario illustrated in Fig. \ref{fig} where two source nodes, ($\mathcal{S}_1$) and ($\mathcal{S}_2$) communicate with each other via an untrusted amplify-and-forward (AF) relay ($\mathcal{R}$) which acts as both a requisite helper as well as a potential passive eavesdropper. Two secure transmission protocols are taken into account: i) WFJ in which one friendly jammer ($\mathcal{FJ}$) is employed to enhance the security of the network by degrading the relay channel through sending its jamming signal, and ii) WoFJ. Note that both $\mathcal{R}$ and $\mathcal{FJ}$ are assumed to be energy-starved nodes, equipped with rechargeable batteries so that they can be wirelessly powered by the sources. It is assumed that most of the nodes' energy are consumed for data transmission, and energy consumption for signal processing is ignored for simplicity \cite{Nasir}. In the WFJ scenario, the data exchange between two sources is implemented in three phases. In the first phase, shown with solid lines,~\src~and~\des~transmit non-information signals to~\jj~and~\rr, to charge them. It is worth mentioning that the helper nodes, $\mathcal{R}$ and $\mathcal{FJ}$, contribute to information transmission until the saved energy of the received signals in the first phase of the communication is more than the specified minimum energy harvesting (EH) threshold. During the second phase, the source nodes send their information signals to the relay. Simultaneously,~\jj~deteriorates the relay channel capacity by transmitting its jamming signal powered by the sources in the first stage, as demonstrated with dashed lines. Finally, in the third phase, the relay broadcasts the scaled version of the received signal to the sources and then they each extract their corresponding information signals after self-interference cancellation. Despite the fact that~\jj~is able to harvest the energy of broadcast signal from~\rr~during the third phase, we do not consider this signal due to its high attenuation experienced through two paths, i.e. from $\mathcal{S}_{1, 2}$ to~\rr, and then, from~\rr~to~\jj. Therefore, its power is negligible to be harvested. 

The WoFJ scenario follows the same three phases as the WFJ scheme assuming the absence of the FJ.

\subsection{Channel Model}

In the proposed scenario, we assume that all the nodes are equipped with a single antenna and operate in half-duplex mode. The direct link between~\src~and~\des~is assumed to be in absence, which is common in the scenarios where two sources are located far away from each other or within heavily shadowed areas  which makes using the relay mandatory {\cite {Kalamkar}}, {\cite {Nasir}}. The channels are assumed to be reciprocal and follow a quasi-static block-fading Rayleigh model \cite{Nasir}. Furthermore, a key assumption is that the sources have perfect knowledge of the jamming signal transmitted by~\jj~as well as the channel state information (CSI) of \src\textendash\rr,~\des\textendash\rr, and~\jj\textendash\rr~channels {\cite {Kuh3}}. Let us denote $h_{ij}$ as the channel coefficient between nodes $i$ and $j$, with channel reciprocity where  $h_{ij}=h_{ji}$. The channel power gain $|h_{ij}|^2$ follows an exponential distribution with mean $\mu_{ij}$ as
\begin{equation}\label{pdf}
f_{|h_{ij}|^2}(x)=\frac{1}{\mu_{ij}}\exp(-\frac{x}{\mu_{ij}}), \quad x>0
\end{equation}
where $f_{|h_{ij}|^2}(x)$ is the probability density function (pdf) of r.v. $|h_{ij}|^2$.

\subsection{Time Switching Relaying Protocol}
\begin{figure}[t]
	\centering  
	\includegraphics[width= \columnwidth]{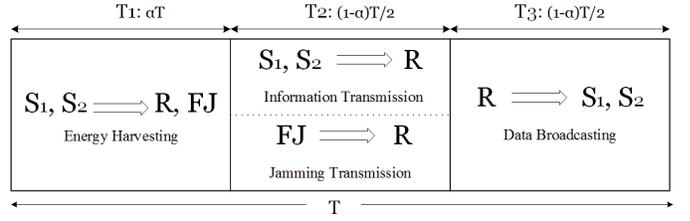}
	\caption{ Time switching relaying protocol for two-way communication via a wireless powered untrusted relay and a friendly jammer}
	\label{TS}
\end{figure}
Fig. \ref{TS} describes the proposed wireless energy harvesting two-way relaying transmission protocol. Using the TS policy, the relay time-switches between information processing and energy harvesting, and completes a round of data exchange in three phases or a period of $T$. To be specific, in the first phase with duration of $T_1=\alpha T$ ($ 0<\alpha<1$), both~\rr~and~\jj~harvest the energy of the RF signals transmitted by~\src~and~\des. In the second time slot which lasts $T_2=(1-\alpha)\frac{T}{2}$,~\src~and~\des~send their information signals to~\rr, while~\jj~transmits its jamming signal, powered by received RF signals during the first phase of communication. Finally, in the third time slot,~\rr~broadcasts the scaled version of the received signal.

\subsection{Energy Harvesting at the Relay and Friendly Jammer}

In the first phase, two source nodes send non-information signals, powering both~\rr~and~\jj~to obtain the required power for activity. The received power at~\rr~and \jj~are respectively, given by \eqref{pr} and \eqref{pj}
\begin{equation}\label{pr}
P_{R}=P_{S_{1}}  |h_{S_{1}R}|^2+P_{S_{2}}  |h_{S_{2}R}|^2,
\end{equation}
\begin{equation}\label{pj}
P_{J}=P_{S_{1}}  |h_{S_{1}J}|^2+P_{S_{2}}  |h_{S_{2}J}|^2.
\end{equation}
Note that $P_{R}$ and $P_{J}$ should be more than the minimum threshold power ($\Theta_{R}$) to activate the harvesting circuitry. In the TS protocol, the harvested energy $E_{HR}$ and $E_{HJ}$ during $\alpha T$ duration at~\rr~and~\jj~are respectively, given by
\begin{equation}\label{ehr}
E_{HR}=\eta_{R} \alpha T(P_{S_{1}}  |h_{S_{1}R}|^2+P_{S_{2}}  |h_{S_{2}R}|^2),
\end{equation}
and
\begin{equation}\label{ehj}
E_{HJ}=\eta_{J}\alpha T(P_{S_{1}}  |h_{S_{1}J}|^2+P_{S_{2}}  |h_{S_{2}J}|^2),
\end{equation}
where $\eta_R$ and $\eta_J$ with $0<\eta_R, \eta_J<1$ represent the energy conversion efficiency factors of~\rr~and~\jj, respectively. The relay uses the harvested energy in \eqref{ehr} to retransmit the source signals in the third phase with power $P_{TR}$ which can be written as
\begin{equation}\label{ptr}
P_{TR}=\frac{E_{HR}}{(1-\alpha)\frac{T}{2}}=\frac{2\eta_{R} \alpha (P_{S_{1}}  |h_{S_{1}R}|^2+P_{S_{2}}  |h_{S_{2}R}|^2)}{1-\alpha}.
\end{equation}
\jj~uses the harvested energy in \eqref{ehj} to transmit its jamming signal with the power of $P_{TJ}$ which can be expressed as
\begin{equation}\label{ptj}
P_{TJ}=\frac{E_{HJ}}{(1-\alpha)\frac{T}{2}}=\frac{2\eta_{J}\alpha (P_{S_{1}}  |h_{S_{1}J}|^2+P_{S_{2}}  |h_{S_{2}J}|^2)}{1-\alpha}.
\end{equation}

\subsection{Information Processing at the Relay}	

The received signal at~\rr, in the second phase, can be written as
\begin{eqnarray}\label{yr}
y_{R}{\hspace {-3mm}}&=&{\hspace {-3mm}}\sqrt{P_{S_{1}}}x_{S_{1}}h_{S_{1}R}+\sqrt{P_{S_{2}}}x_{S_{2}}h_{S_{2}R} \nonumber \\
{\hspace {-3mm}}&+&{\hspace {-3mm}}\sqrt{P_{TJ}}x_{J}h_{JR}+n_{R},
\end{eqnarray}
where $x_{S_{1}}$, $x_{S_{2}}$ and $x_{J}$ are the~\src,~\des~and~\jj~signals with the powers of $P_{S_{1}}$, $P_{S_{2}}$ and $P_{TJ}$, respectively. Moreover, $n_{R}$ is considered as the additive white Gaussian noise (AWGN) at the relay and for simplicity the processing noise at the relay is ignored {\cite {Kalamkar}}. Note that ~\jj~sends its jamming signal $x_{J}$ with total power harvested in the first phase, which is higher that the minimum threshold power for circuitry activation. Based on the received signal $y_{R}$ in \eqref{yr}, the signal-to-interference-plus-noise-ratio at $\mathcal{R}$ can be expressed as

\begin{eqnarray} \label{gammaR}
\gamma_{R}{\hspace {-3mm}}&=&{\hspace {-3mm}}\frac{
	P_{S_{1}}  |h_{S_{1}R}|^2+P_{S_{2}}  |h_{S_{2}R}|^2}{P_{TJ}  |h_{JR}|^2+N_{0}} \nonumber \\
{\hspace {-3mm}}&=&{\hspace {-3mm}}\mbox{}\frac{P_{S_{1}}  |h_{S_{1}R}|^2+P_{S_{2}}  |h_{S_{2}R}|^2}{\frac{2\eta_{J}\alpha}{1-\alpha} (P_{S_{1}}  |h_{S_{1}J}|^2+P_{S_{2}}  |h_{S_{2}J}|^2) |h_{JR}|^2+N_{0}},
\end{eqnarray}
where $N_{0}$ is the noise power at $\mathcal{R}$ and we assume that the relay performs multiuser decoding to estimate the signals from~\src~and~\des. Finally, $\mathcal{R}$ broadcasts
\begin{eqnarray}\label{output_relay}
x_{R}=\zeta y_{R},
\end{eqnarray}
where the scaling factor of $\mathcal{R}$ is
\begin{equation}\label{zeta}
\zeta=\sqrt{\frac{P_{TR}}{P_{S_{1}}  |h_{S_{1}R}|^2+P_{S_{2}}  |h_{S_{2}R}|^2+P_{TJ}|h_{JR}|^2+N_{0}}}.
\end{equation}

\subsection{Information Processing at the Sources}	

Next, we detail the received signal at~\des, from which similar expressions can be derived for the received signal at~\src. By using \eqref{yr} and \eqref{output_relay}, the received signal at~\des~can be expressed as
\begin{eqnarray}\label{ys2prim}
y_{S_{2}}'{\hspace {-3mm}}&=&{\hspace {-3mm}}h_{RS_{2}}x_{R}+n_{S_{2}} \nonumber \\
{\hspace {-3mm}}&=&{\hspace {-3mm}}\sqrt{P_{S_{1}}} \zeta h_{S_{1}R}h_{RS_{2}}x_{S_{1}}+\sqrt{P_{S_{2}}} \zeta h_{S_{2}R}h_{RS_{2}}x_{S_{2}} \nonumber\\
{\hspace {-3mm}}&+&{\hspace {-3mm}}\sqrt{P_{TJ}}\zeta h_{JR}h_{RS_{2}}x_{J}+\zeta h_{RS_{2}}n_{R}+n_{S_{2}},
\end{eqnarray}
where $n_{S_{2}}$ is the AWGN at~\des~with the power $N_{0}$. Since we assume that~\src~and~\des~know the jamming signal ($x_{J}$) transmitted by~\jj~and the CSI of~\src$\leftrightarrow$\rr,~\des$\leftrightarrow$\rr~and~\rr$\leftrightarrow$\jj, this means that~\src~and~\des~are able to cancel the self-interference and the jamming signals in \eqref{ys2prim}. Accordingly, the received signal at~\des~can be simplified as
\begin{equation}\label{ys2}
y_{S_{2}}=\sqrt{P_{S_{1}}} \zeta h_{S_{1}R}h_{RS_{2}}x_{S_{1}}+\zeta h_{RS_{2}}n_{R}+n_{S_{2}}.
\end{equation}
Substituting \eqref{zeta} into \eqref{ys2}, the received instantaneous end-to-end SNR at~\des~can be obtained as
\begin{equation} \label{gammad1hsnr}
\gamma_{S_{2}}=\frac{2\eta_{R}\alpha P_{S_{1}}  |h_{S_{1}R}|^2 |h_{RS_{2}}|^2}{{2\eta_{R}\alpha|h_{RS_{2}}|^2 N_{0}+
		\frac{N_{0}(P_{TJ}|h_{JR}|^2(1-\alpha))}{P_{S_{1}}|h_{S_{1}R}|^2+P_{S_{2}}|h_{S_{2}R}|^2}}
	{+N_{0}(1-\alpha)+\epsilon}}.
\end{equation}
where $\epsilon=\frac{N_{0}^2(1-\alpha)}{P_{S_{1}}|h_{S_{1}R}|^2+P_{S_{2}}|h_{S_{2}R}|^2}$. Following the same procedure of~\des, the received instantaneous end-to-end SNR at~\src~is given by
\begin{equation} \label{gammad2hsnr}
\gamma_{S_{1}}=\frac{2\eta_{R}\alpha P_{S_{2}}  |h_{S_{2}R}|^2 |h_{RS_{1}}|^2}
{{2\eta_{R}\alpha|h_{RS_{1}}|^2 N_{0}+
		\frac{N_{0}(P_{TJ}|h_{JR}|^2(1-\alpha))}{P_{S_{1}}|h_{S_{1}R}|^2+P_{S_{2}}|h_{S_{2}R}|^2}}
	{+N_{0}(1-\alpha)+\epsilon}}.
\end{equation}
To make the subsequent analysis tractable, we proceed to examine high SNR relaying regime where $\epsilon=0$ in \eqref{gammad1hsnr} and \eqref{gammad2hsnr}.

\section{Performance Analysis}
In this section, we first derive closed-form expressions for the power outage probability at the relay and the FJ. Then, analytical expressions and closed-form lower bound expressions are evaluated for both the cases of WFJ and WoFJ.

\subsection{Power Outage Probability at the Relay}
The received power at the relay must be greater than the minimum power threshold $\Theta_{R}$ to activate the energy harvesting circuitry. If the received power $P_{R}$, in \eqref{pr}, is less than the threshold power $\Theta_{R}$, the power outage probability occurs, which is defined as
\begin{equation}\label{pop}
P_{por}=\ensuremath{\mathbb P}\left( P_{R} < \Theta_{R} \right)
\end{equation}
The power outage probability at the relay is computed using the following proposition.

\begin{proposition}
 The power outage probability at the relay is calculated as
 \begin{equation} \label{por}
 P_{por}=\left \{\begin{array}{ll} 1-\frac{P_{S_{2}}\mu_{S_{2}R}}{P_{S_{2}}\mu_{S_{2}R}-P_{S_{1}}\mu_{S_{1}R}}\exp(-\frac{\Theta_{R}}{P_{S_{2}}\mu_{S_{2}R}})\\ +\frac{P_{S_{1}}\mu_{S_{1}R}}{P_{S_{2}}\mu_{S_{2}R}-P_{S_{1}}\mu_{S_{1}R}}	\exp(-\frac{\Theta_{R}}{P_{S_{1}}\mu_{S_{1}R}}),\space m_x \neq m_y \\     \Upsilon(2,m\Theta_{R}),\qquad \qquad \qquad \qquad m_x = m_y\treq{m} 	\end{array}
 \right.
 \end{equation}
 where $m_x{\hspace {-1mm}}={\hspace {-1mm}}\frac{1}{P_{S_{1}}\mu_{S_{1}R}}$, $m_y{\hspace {-1mm}}={\hspace {-1mm}}\frac{1}{P_{S_{2}}\mu_{S_{2}R}}$, and $\Upsilon(s,x){\hspace {-1mm}}={\hspace {-1mm}}\int_{0}^{x}t^{(s-1)}\emph{e}^{-t} dt$ is the lower incomplete Gamma function \cite{papoulis}.
\end{proposition}
\begin{proof}
See Appendix A.
\end{proof}
\begin{remark}
The power outage probability expression at~\jj~($P_{poj}$) can be obtained similar to $P_{por}$, by replacing $\mu_{S_1R}=\mu_{S_1J}$ and $\mu_{S_2R}=\mu_{S_2J}$ in \eqref {por}.
\end{remark}

\subsection{Ergodic Secrecy Sum Rate}
The ergodic secrecy rate characterizes the rate below which the average secure communication transmission is not achievable {\cite {review0}}. For our proposed system, the instantaneous secrecy rate $R_{sec}$  can be expressed as \cite{Kuh3}
\begin{equation}\label{rsec1}
R_{sec}=\left[I_{S_{1}}+I_{S_{2}}-I_{R}\right]^+,
\end{equation}
where
\begin{equation}\label{rsec2}
I_{K}=\frac{(1-\alpha)}{2}\log_{2}(1+\gamma_{K}),
\end{equation}
and $K~\in~$\{\src, \des, \rr\}. By combining \eqref{rsec1} and \eqref{rsec2}, the ESSR can be rewritten as
\begin{equation}\label{issr}
R_{sec}=\left[\frac{(1-\alpha)}{2}\log_{2}\frac{(1+\gamma_{S_{1}})(1+\gamma_{S_{2}})}{(1+\gamma_{R})}\right]^+,
\end{equation}
where $[x]^+ = \max(x,0)$ and the pre-log factor  $\frac{1-\alpha}{2}$ is due to the effective time of information exchange between the two sources. Moreover, $\gamma_{R}$, $\gamma_{S_{1}}$ and $\gamma_{S_{2}}$ are given by \eqref{gammaR}, \eqref{gammad1hsnr} and \eqref{gammad2hsnr}, respectively.

Based on \eqref{issr} and the relay and FJ's power outage probability, the exact expression of the ESSR can be written as
\begin{equation} \label{esssr}
{\bar{R}_{sec}=(1-P_{poj})(1-P_{por})\E\left\lbrace {R_{sec}}\right\rbrace }.
\end{equation}
which can be further expressed as
\begin{eqnarray}\label{essr}
\bar{R}_{sec}{\hspace {-3mm}}&=&{\hspace {-3mm}}(1-P_{poj})(1-P_{por}) \nonumber \\ {\hspace {-3mm}}&\times&{\hspace {-3mm}}\int_{0}^{\infty}\int_{0}^{\infty}\int_{0}^{\infty}\int_{0}^{\infty}\int_{0}^{\infty} R(x,y,z,w,u) \nonumber \\ {\hspace {-3mm}}&\times&{\hspace {-3mm}}f_{X}(x)f_{Y}(y)f_{Z}(z)f_{U}(u)f_{W}(w)dxdydzdudw,
\end{eqnarray}
where
\begin{equation}
R(x,y,z,w,u)=\left[\frac{(1-\alpha)}{2}\log_{2}\frac{(1+\gamma_{S_{2}})(1+\gamma_{S_{1}})}{(1+\gamma_{R})}\right]^+,
\end{equation}
where we define $X{\hspace{-1mm}}={\hspace{-1mm}}|h_{S_{1}R}|^2$, $Y{\hspace{-1mm}}={\hspace{-1mm}}|h_{S_{2}R}|^2$, $Z{\hspace{-1mm}}={\hspace{-1mm}}|h_{S_{1}J}|^2$, $W{\hspace{-1mm}}={\hspace{-1mm}}|h_{S_{2}J}|^2$ and $U{\hspace{-1mm}}={\hspace{-1mm}}|h_{RJ}|^2$ in the r.v.s of $\gamma_{R}$, $\gamma_{S_{1}}$ and $\gamma_{S_{2}}$.

While the multiple integral expression in \eqref{essr} can be easily evaluated numerically, a closed-form expression is not straightforward to obtain. As such, we proceed by obtaining a closed-form lower bound expression for the ESSR.

\begin{proposition}
The lower bound expression $\bar{R}_{LB}$ for the ergodic secrecy sum rate $\bar{R}_{sec}$  given by \eqref{esssr} can be stated as	
\begin{equation}\label{R_LB}
\bar{R}_{LB}=(1-P_{poj})(1-P_{por}) \frac{1-\alpha}{2\ln(2)}\left[T_1+T_2-T_3\right]^+,
\end{equation}
where
\begin{eqnarray} \label{T_1}
T_{1}\approx\ln\hspace*{-1mm}\left(\hspace*{-1mm}1+\hspace*{-1mm}\frac{2\eta_{R}\alpha P_{S_{1}}\mu_{S_{2}R}\mu_{S_{1}R}}
{\splitdfrac{2\eta_{R}\alpha\mu_{S_{2}R}N_{0}+N_{0}(1-\alpha)+2N_{0}\eta_{J}}
	{\hspace*{-15mm}\times\alpha\mu_{RJ}\left(P_{S_{1}}\mu_{S_{1}J}+P_{S_{2}}\mu_{S_{2}J}\right)\E\left\lbrace H\right\rbrace}}
\right),
\end{eqnarray}
\begin{eqnarray}\label{T_2}
T_{2}\approx\ln\hspace*{-1mm}\left(\hspace*{-1mm}1+\hspace*{-1mm}\frac{2\eta_{R}\alpha P_{S_{2}}\mu_{S_{2}R}\mu_{S_{1}R}}
{\splitdfrac{2\eta_{R}\alpha\mu_{S_{1}R}N_{0}+N_{0}(1-\alpha)+2N_{0}\eta_{J}}
	{\hspace*{-15mm}\times\alpha\mu_{RJ}\left(P_{S_{1}}\mu_{S_{1}J}+P_{S_{2}}\mu_{S_{2}J}\right)\E\left\lbrace H\right\rbrace}}
\right),
\end{eqnarray}
\begin{eqnarray}\label{T_3}
T_{3}\approx\ln\left(1+\frac{P_{S_{1}}\mu_{S_{1}R}+P_{S_{2}}\mu_{S_{2}R}}{\frac{2\eta_{J}\alpha}{1-\alpha}\left(P_{S_{1}}\mu_{S_{1}J}+P_{S_{2}}\mu_{S_{2}J}\right)\mu_{RJ}+N_0} \right),
\end{eqnarray}
and
\begin{eqnarray}\label{Eh}
\E\left(H\right)=\left \{ \begin{array}{ll}
{ \frac{\ln  \left( {\frac {\mu _{R}}{\mu _{S}}} \right)}{{\mu _{R}-\mu _{S}}}}
,&  \quad\mu_{S} \neq \mu_{R}  \\
\frac{1}{\mu}. & \quad \mu_{S} = \mu_{R} \treq \mu
\end{array} 	
\right.
\end{eqnarray}
\end{proposition}
\begin{proof}
See Appendix B.
\end{proof}
\begin{remark}
As shown in the numerical results, the lower bound expression in \eqref{R_LB} is tight enough in high SNR regime.
\end{remark}

\section{Numerical Results and Discussions}
In this section, we provide some numerical results to verify the accuracy of the provided expressions. We set $P_{S_{1}}=P_{S_{2}}=10$ dBW, $\eta_R=\eta_J=0.7$, the minimum EH circuitry threshold $\Theta_R=0$ dBm, the noise power $N_{0}=-10$ dBm, $d_{S_1R}=d_{S_2R}=d_{S_1J}=d_{S_2J}=2d_{RJ}=3$m, the mean channel power gain $\mu_{ij}=d_{ij}^{-\rho}$ and the path-loss exponent $\rho = 2.7$ similar to \cite{Lu}.

\begin{figure}[t]
	\centering
	\includegraphics[width= \columnwidth]{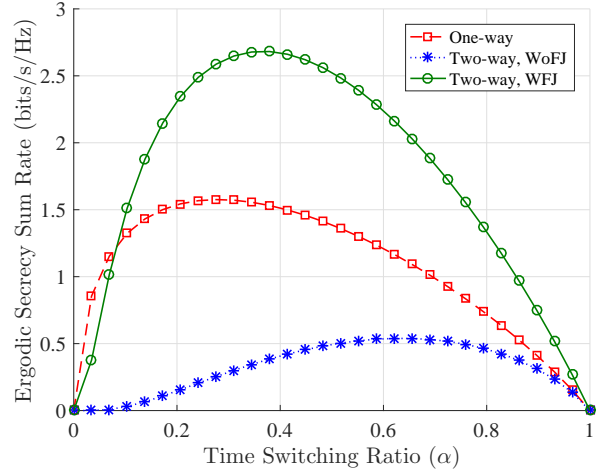}
	\caption{\small ESSR versus the time switching ratio of one-way, two-way WoFJ, and two-way WFJ scenarios.}
	\label{fig3}
\end{figure}

Fig. \ref{fig3}  shows that the ESSR is a quasi-concave function of the time switching ratio. For the given system parameters, the maximum ESSR are obtained at $\alpha_{opt} = 0.28$, $\alpha_{opt} = 0.38$ and $\alpha_{opt} = 0.62$, for one-way, two-way WFJ and two-way WoFJ, respectively. This highlights the importance of the TS ratio which should be taken into account in the system design. This observation means that the security of the network is highly dependent on the both jamming strategy (WFJ or WoFJ) and the TS ratio. If the TS ratio is too low, the harvested energy at the relay (and the FJ) may be too low and then, power outage may occur or the received SNR at the sources may be too low. On the other hand, if the TS ratio is too high, insufficient time is dedicated for the relay to broadcast the information signal and hence, the received instantaneous end-to-end SNR at the receivers may be too low. As such, the figure reveals a trade-off between a secure transmission and a reliable communication which will be taken into account in our future work.

\begin{figure}[t]
	\centering
	\includegraphics[width= \columnwidth]{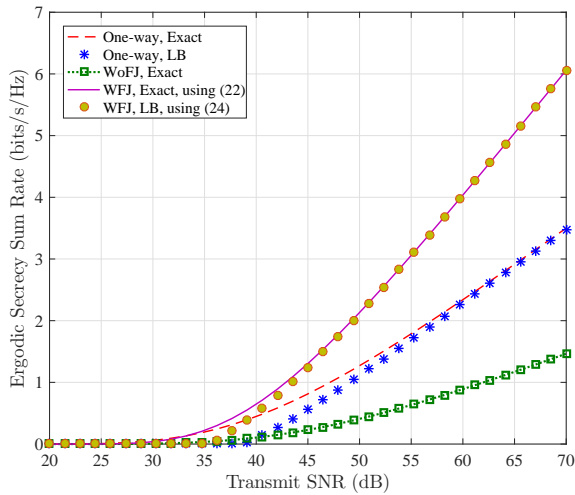}
	\caption{\small ESSR versus the transmit SNR of one-way, two-way WoFJ, and two-way WFJ scenarios. }
	\label{fig4}
\end{figure}

Fig. \ref{fig4} plots the ESSR versus transmit SNR of one-way, two-way WFJ, and two-way WoFJ scenarios. We observe that: 1) The exact numerical expression in \eqref{essr} is well-approximated in the high SNR regime by our derived closed-form lower bound expression in \eqref{R_LB}, 2) The ESSR is significantly enhanced as the transmit SNR increases for both the one-way scenario \cite{Kalamkar} and the proposed two-way communication scenarios, 3) In the high SNR regime the proposed two-way WFJ communication  outperforms the one-way scenario. For example, at SNR = 50 dB, the ESSR of the two-way WFJ scenario is 1 bit/s/Hz more than that the one-way transmission scenario provides; approximately two times as much as it does, and 4) Evidently, the high SNR-slope of the proposed two-way WFJ communication is twice the one-way scenario and the two-way WoFJ scenario.

\begin{figure}[t]
	\centering
	\includegraphics[width= \columnwidth]{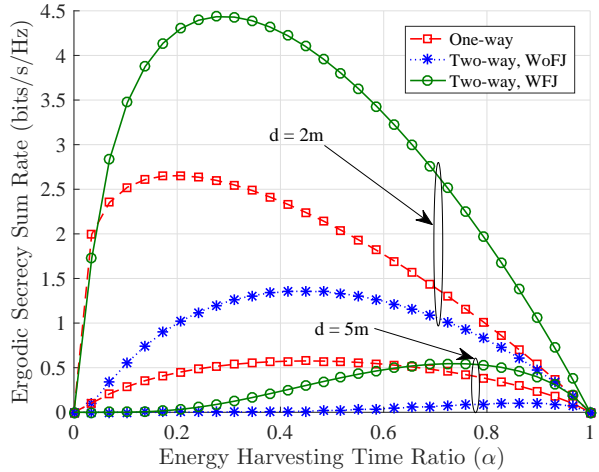}
	\caption{\small ESSR versus the energy harvesting time ratio of one-way, two-way WoFJ, and two-way WFJ scenarios with different distances between nodes. }
	\label{fig5}
\end{figure}

Fig. \ref{fig5} illustrates the impact of distances between the network nodes on the ESSR. We assume that all the nodes except~\src-\des~are located in equal distances from each other denoted by $d$. One interesting result from Fig. \ref{fig5} is that the proposed two-way WFJ scenario significantly outperforms the one-way and two-way WoFJ communications when the network nodes are close together, i.e. $d=2$m. It is also observed that in this topology, the energy harvesting time is much less than the broadcasting time to achieve the optimal ESSR. Moreover, by extending the network scale to $d=5$m, the maximum ESSR can be achieved if more time is dedicated to RF energy harvesting. In other words, as observed in Fig. \ref{fig5}, for the given system parameters of the WFJ scenario, the maximum ESSR with respect to $\alpha$ is achieved when we set $\alpha=0.75$. Consequently, to achieve the maximum ESSR of the proposed WFJ scenario, the more time is required to harvest RF energy in the first phase as the network topology extends.

\section{Conclusion}
In this paper, we proposed a wirelessly powered two-way cooperative network in which two sources communicate via an untrusted relay. To enhance the secrecy rate, a friendly jammer (FJ) is employed. We adapted the time switching (TS) protocol at the relay and investigated the ergodic secrecy sum rate (ESSR) with friendly jamming (WFJ). An accurate lower bound expression was derived for the ESSR. Numerical results revealed the priority of the proposed two-way communication. Furthermore, novel engineering insights were discussed regarding the impact of TS ratio, SNR and node locations on the ESSR.

\section*{Appendix A}
The power outage probability at~\rr~is represented by
\begin{equation*}\label{wsx}
P_{por}=\ensuremath{\mathbb P}\left( P_{R} < \Theta_{R} \right),
\end{equation*}
where $P_{R}=X+Y$, and
$X=P_{S_{1}}|h_{S_{1}R}|^2$ and $Y=P_{S_{2}}|h_{S_{2}R}|^2$ are two exponential r.v.s with means equal to $\frac{1}{m_x}=P_{S_{1}}\mu_{S_{1}R}$ and $\frac{1}{m_y}=P_{S_{2}}\mu_{S_{2}R}$, respectively. Accordingly, the probability density function of $S$ can be obtained as
\begin{eqnarray}\label{qaz}
f_S(s){\hspace {-3mm}}&=&{\hspace {-3mm}}\int_{0}^{s} f_{XY}(x,s-x)dx \nonumber \\
{\hspace {-3mm}}&\stackrel{(a)}{=}&{\hspace {-3mm}} \int_{0}^{s} f_{X}(x)f_{Y}(s-x)dx \nonumber \\
{\hspace {-3mm}}&=&{\hspace {-3mm}}\int_{0}^{s}m_x\emph{e}^{-m_xx}m_y\emph{e}^{-m_y(s-x)}dx \nonumber \\
{\hspace {-3mm}}&=&{\hspace {-3mm}}	\left\{\begin{array}{ll}
m^2s\emph{e}^{-ms},  \space\quad \qquad \qquad m_x=m_y\treq m \\
\frac{m_xm_y}{m_x-m_y}\left(\emph{e}^{-m_xs}-\emph{e}^{-m_ys}\right), \quad m_x\neq m_y
\end{array}
\right.
\end{eqnarray}
where $s$ can take only non-negative values and $(a)$ follows from the fact that two r.v.s X and Y are independent. Then, using \eqref{qaz}, we can write
\begin{equation} \label{qwe}
P_{por}=\ensuremath{\mathbb P}\left( S < \Theta_{R} \right) =\int_{0}^{\Theta_{R}}f_S(s)ds.
\end{equation}
By evaluating the integral expression \eqref{qwe}, we arrive at the $P_{por}$ in \eqref{por}.

\section*{Appendix B}

The lower bound expression of the ESSR $\bar{R}_{sec}$ can be written as follows
\begin{eqnarray}\label{lb}
\bar{R}_{sec}{\hspace {-3mm}}&=&{\hspace {-3mm}}(1-P_{poj})(1-P_{por}) \nonumber \\ {\hspace {-3mm}}&\times&{\hspace {-3mm}}\E\left \lbrace  \frac{(1-\alpha)}{2}\left[\log_{2}\frac{(1+\gamma_{S_{2}})(1+\gamma_{S_{1}})}{(1+\gamma_{R})}\right]^+\right\rbrace  \\ {\hspace {-3mm}}&\stackrel{(a)}{\geq}&{\hspace {-3mm}}(1-P_{poj})(1-P_{por}) \nonumber\\ {\hspace {-3mm}}&\times&{\hspace {-3mm}}\left[ \E\left\lbrace  \frac{(1-\alpha)}{2} \left(\log_{2}\frac{(1+\gamma_{S_{2}})(1+\gamma_{S_{1}})}{(1+\gamma_{R})}\right)\right\rbrace \right]^+ \nonumber \\ {\hspace {-3mm}}&=&{\hspace {-3mm}}(1-P_{poj})(1-P_{por}) \times\nonumber\\ {\hspace {-3mm}}&&{\hspace {-3mm}} \begin{split} &\biggm[\frac{1-\alpha}{2\ln(2)}\bigg( \stackrel{}{\underset{T_{1}}{\underbrace{\E\left\lbrace\ln\left(1+\gamma_{S_{2}}\right)\right\rbrace}}}+\stackrel{}{\underset{T_{2}}{\underbrace{\E\left\lbrace\ln\left(1+\gamma_{S_{1}}\right)\right\rbrace}}}  \\ &-\stackrel{}{\underset{T_{3}}{\underbrace{\E\left\lbrace\ln\left(1+\gamma_{R}\right)\right\rbrace }}}\bigg)\biggm]^+\treq\bar{R}_{LB}, \\ \end{split}
\end{eqnarray}
where inequality $(a)$ follows from the fact that $\E\left\lbrace \max(X,Y)\right\rbrace {\hspace {-2mm}}\geq{\hspace {-2mm}} \max(\E\{X\},\E\{Y\})$ \cite{papoulis}. Moreover, for any positive r.v.s X and Y, the following approximation can be used \cite{A} \[\E\left\lbrace\ln\left( 1+\frac{X}{Y}\right) \right\rbrace \approx \ln\left( 1+\frac{\E\left\lbrace X\right\rbrace }{\E\left\lbrace Y\right\rbrace }\right). \]
Based on the above approximation and using the expression \eqref{gammad1hsnr}, the part $T_{1}$ is given by
\begin{eqnarray}\label{x}
T_{1}\approx\ln\left(1+\frac{\E \left\lbrace 2\eta_{R}\alpha P_{S_{1}}XY\right\rbrace } {\E\left\lbrace\splitdfrac{2\eta_{R}\alpha YN_{0}+N_{0}(1-\alpha)}{+\frac{2N_{0}\eta_{J}\alpha(P_{S_{1}}Z+P_{S_{2}}W)U}{P_{S_{1}}X+P_{S_{2}}Y}}\right\rbrace }\right).
\end{eqnarray}
Since the r.v.s X, Y, Z, W and U are independent, by using the expression \eqref{gammad2hsnr}, we can conclude \eqref{T_1}. $T_2$ and $T_3$ can be calculated in the same way as $T_1$, and \eqref{T_2} and \eqref{T_3} are obtained. Finally, to evaluate $\E(H)$, we define $H\hspace{-1mm}=\hspace{-1mm}\frac{1}{P_{S_{1}}X+P_{S_{2}}Y}$ as a new r.v. Let us denote $R=P_{S_{1}}X$, $ S=P_{S_{2}}Y$ as two exponential r.v.s with the means equal to  $\mu_{R}=P_{S_{1}}\mu_{S_{1}R}$ and $\mu_{S}=P_{S_{2}}\mu_{S_{2}R}$, respectively. Using \eqref{qaz}, $f_{H}(h)$ can be calculated as
\begin{align}\label{fH}
f_{H}(h)=\begin{cases}
\frac{ \exp\left(-\frac{1}{h\mu_{S}}\right)-\exp\left( -\frac{1}{h\mu_{R}}
	\right)}{h^{2}\left( \mu_{S}-\mu_{R}\right) }
, &  \mu_{R}\neq\mu_{S}\\
\frac{1}{h\mu^2}\exp(-\frac{1}{h\mu}), &  \mu_{R}=\mu_{S}\treq\mu
\end{cases}
\end{align}
By computing $\E(H)=\int_{0}^{\infty}hf_{H}(h)dh$, \eqref{Eh} is proved.

\end{document}